\documentclass[10pt,onecolumn]{IEEEtran}
\usepackage{amsmath,amssymb,euscript ,yfonts,psfrag,latexsym,dsfont,graphicx,bbm,color,amstext,wasysym,subfig,flushend,parskip}
\graphicspath{{./},{./figures/}}

\begin{document}
\newtheorem{thm}{Theorem}
\newtheorem{cor}[thm]{Corollary}
\newtheorem{conj}[thm]{Conjecture}
\newtheorem{lemma}[thm]{Lemma}
\newtheorem{prop}{Proposition}
\newtheorem{problem}[thm]{Problem}
\newtheorem{remark}[thm]{Remark}
\newtheorem{defn}[thm]{Definition}
\newtheorem{ex}[thm]{Example}

\newcommand{\mR}{{\mathbb R}}
\newcommand{\mD}{{\mathbb D}}
\newcommand{\E}{{\mathbb E}}  
\newcommand{\cN}{{\mathcal N}}
\newcommand{\cR}{{\mathcal R}}
\newcommand{\cS}{{\mathcal S}}
\newcommand{\cC}{{\mathcal C}}
\newcommand{\diag}{\operatorname{diag}}
\newcommand{\tr}{\operatorname{trace}}
\newcommand{\f}{{\mathfrak f}}
\newcommand{\g}{{\mathfrak g}}
\newcommand{\range}{\cR}
\newcommand{\rH}{{\rm H}}
\newcommand{\trace}{\operatorname{trace}}
\newcommand{\argmin}{\operatorname{argmin}}

\newcommand{\ignore}[1]{}

\def\spacingset#1{\def\baselinestretch{#1}\small\normalsize}
\setlength{\parskip}{10pt}
\setlength{\parindent}{20pt}
\spacingset{1}

\newcommand{\mike}{\color{magenta}}
\definecolor{grey}{rgb}{0.6,0.6,0.6}
\definecolor{lightgray}{rgb}{0.97,.99,0.99}

\title{Optimal steering of a linear stochastic system\\ to a final probability distribution, part II
}

\author{Yongxin Chen, Tryphon Georgiou and Michele Pavon
\thanks{Y.\ Chen and T.T.\ Georgiou are with the Department of Electrical and Computer Engineering,
University of Minnesota, Minneapolis, Minnesota MN 55455, USA; {email: \{chen2468,tryphon\}@umn.edu}}
\thanks{M.\ Pavon is with the Dipartimento di Matematica,
Universit\`a di Padova, via Trieste 63, 35121 Padova, Italy; {email: pavon@math.unipd.it}}}
\markboth{October 13, 2014}{}

\maketitle
{\begin{abstract}
We consider the problem of minimum energy steering of a linear stochastic system to a final prescribed distribution over a finite horizon and the problem to maintain a stationary distribution over an infinite horizon. For both problems the control and noise channels are allowed to be distinct, thereby, placing the results of this paper outside of the scope of previous work both in probability and in control.
We present sufficient conditions for optimality in terms of a system of dynamically coupled Riccati equations in the finite horizon case and in terms of algebraic conditions for the stationary case.
We then address the question of {\em feasibility} for both problems. For the finite-horizon case, provided the system is controllable, we prove that without any restriction on the directionality of the stochastic disturbance it is always possible to steer the state to any arbitrary Gaussian distribution over any specified finite time-interval.
For the stationary infinite horizon case, it is not always possible to maintain the state at an arbitrary Gaussian distribution through constant state-feedback. It is shown that covariances of admissible stationary Gaussian distributions are characterized by a certain Lyapunov-like equation and, in fact, they coincide with the class of stationary state covariances that can be attained by a suitable stationary colored noise as input.
We finally address the question of how to compute suitable controls numerically.
We present an alternative to solving the system of coupled Riccati equations,
by expressing the optimal controls in the form of solutions to (convex) semi-definite programs for both cases.
We conclude with an example to steer the state covariance of the distribution of inertial particles to an admissible stationary Gaussian distribution over a finite interval, to be maintained at that stationary distribution thereafter by constant-gain state-feedback control.
\end{abstract}}


\section{Introduction}
Consider a linear system
\begin{equation}\label{eq:linearsystem}
\dot x(t)=Ax(t)+Bu(t), \;t\in[0,\infty)
\end{equation}
with $A\in \mR^{n\times n}$, $B\in\mR^{n\times m}$, $x(t)\in\mR^n$ and $u(t)\in\mR^m$,
and the problem to steer \eqref{eq:linearsystem} from the origin to a given point $x(T)=\xi\in\mR^n$. This of course is possible for any arbitrary $\xi\in\mR^n$ iff the system is {\em controllable}, i.e.,
the rank of $[B,\,AB,\ldots,\,A^{n-1}B]$ is $n$, that is, when $(A,B)$ is a {\em controllable pair}. In this case it is well known that the steering can be effected in a variety of ways, including ``minimum-energy'' control, over any prespecified interval $[0,T]$. On the other hand, the problem to achieve and maintain a fixed value $\xi$ for the state vector in a stable manner is not always possible. For this to be the case for a given $\xi$, using feedback and feedforward control, the equation
\begin{equation}\label{eq:steadystate}
0=(A-BK)\xi + Bu
\end{equation}
must have a solution $(u,K)$ for a constant value for the input $u$ and a suitable value of $K$ so that $A-BK$ is  Hurwitz (i.e., the feedback system be asymptotically stable).
It is easy to see that this reduces simply to the requirement that $\xi$ satisfies the equation
\[
0=A\xi+Bv
\]
for some $v$; if there is such a $v$, we can always choose a suitable $K$ so that $A-BK$ is Hurwitz and then, from $v$ and $K$, we can compute the constant value $u$. Conversely, from $u$ and $K$ we can obtain $v=u-K\xi$.

In the present paper, we discuss an analogous and quite similar dichotomy between our ability to assign the state-covariance of a linear stochastically driven system by steering the system over an interval $[0,T]$, and our ability to assign the state-covariance of the ensuing stationary state process through constant state-feedback. It will be shown that the state-covariance can be assigned at the end of an interval through suitable feedback control if and only if the system is controllable. On the other hand, a positive semidefinite matrix is an admissible stationary state-covariance attained through constant feedback if and only if it satisfies a certain Lyapunov-like algebraic equation. Interestingly, the algebraic equation that specifies which matrices are admissible stationary state-covariances through constant feedback is the same equation that characterizes stationary state-covariances attained through colored stationary input noise in open loop.

Both of these problems, to steer and possibly maintain the state statistics of a stochastically driven system, are motivated by technological advances that are now available to manipulate {micro} thermodynamic systems and to measure physical properties with unprecedented accuracy.  These include thermally driven atomic force microscopy \cite{toyabe2010nonequilibrium,gannepalli2005thermally}, molecular motors/ratchets, manipulation of macromolecules, laser tweezers \cite{toyabe2010nonequilibrium,braiman2003control,hayes2001active}, and the emergence of very high resolution measuring apparatuses \cite{rowan2000gravitational}. The relevant applications can be exemplified by the need to limit state-uncertainty of linear oscillators that are coupled to a heat bath using feedback, and by the need to control the collective behavior of a swarm of inertial particles experiencing stochastic forcing. The first type of application is encountered, for instance, in the micromechanical systems, laser driven reactions and, in particular, the active cooling of oscillators in devices aimed at measuring, e.g., gravitational waves \cite{ricci2014low}. The second type of application can be captured by the need to focus particle beams \cite{petroni2000stochastic}.

Historically, the problem to steer the probability density of Brownian particles in their path across two points in time, has its origin in a study published in 1931/1932 by Erwin Schr\"odinger \cite{schrodinger1931umkehrung},  \cite[Section VII]{schrodinger1932theorie}. In this, Schr\"odinger asked for the most likely trajectory of particles that are observed, at the two end points of their path, to be distributed according to given empirical distributions. The answer he gave, which provides an updated probability law on path space, in fact relates to {a minimum energy stochastic control} problem \cite{dai1991stochastic}. The subject, which advanced with leaps and bounds over the past 80 years by contributions from Fortet, Beurling, Jamison, F\"ollmer, and many others, has come to be known as Schr\"odinger bridges. Yet, all prior work, was related to the case where the diffusive particles are modeled by non-degenerate diffusions where the noise affects directly all entries of (vectorial) stochastic process, and the link to minimum-energy optimal control was drawn primarily via the Girsanov transformation \cite{dai1991stochastic} for that case.
Recent attempts to address linear stochastic systems were also limited to non-degenerate diffusions {where the control and noise channel are identical} \cite{beghi1997continuous,vladimirov2012minimum}.

In a ``sister paper'' that preceeds the present one \cite{chen2014optimal}, we presented a theory of Schr\"odinger bridges for general linear stochastic systems. This includes possibly degenerate linear diffusions and the theory entails two coupled homogeneous differential Riccati equations, in the style of classical LQR theory, which however are nonlinearly coupled through boundary conditions. For this case, where the equations are {\em only} coupled through their boundary conditions, it is shown in \cite{chen2014optimal} that they can be solved in {\em closed form} for the minimum-energy control. Interestingly, the development falls outside standard LQR theory because, aside for the coupling, the boundary conditions of the Riccati equations are in general {\em sign indefinite}. The salient feature of classical Schr\"odinger bridges, and that of the theory in our ``sister manuscript'' \cite{chen2014optimal}, is that the control which provides the needed drift to reconcile the empirical marginals {\em enters along { the same ``directions'' that the noise affects}, that is, control and noise channels are {\em identical}.}
The present work departs from \cite{chen2014optimal} in that {\em control and noise channels may now differ}.
Hence, no assumption on the directionality of our control authority as compared to that of the random driving noise is being made.
{ Interestingly, while certain aspects parallel \cite{chen2014optimal} (e.g., variational analysis, cf.\ Section \ref{sec:variational2}), the techniques needed to dertermine our ability to steer the state statistics and determine the corresponding control input are quite different.}

{The structure of the paper is as follows: In Section \ref{sec:variational2} we formulate both the finite horizon problem and the infinite horizon stationary problem, and present sufficient conditions for optimality.
In Section~\ref{sec:feasibility1} we consider the  {\em feasibility} of steering the statistics over a finite interval by a suitable control action and
in Section~\ref{sec:feasibility2} we consider the possibility to maintain stationary state-statistics by constant state-feedback.
In Sections \ref{sec:num} and \ref{sec:mestationary}, we formulate the least-energy optimal control problem in each of the two cases, finite horizon and stationary statistics, as semidefinite programs.
Finally, Section \ref{sec:example} highlights the theory with a numerical example to steer the statistics of inertial particles, in the phase-plane, in each of these two modalities, transient and stationary.
}

\section{Optimal steering}\label{sec:variational2}

In this section we formulate the control problem to optimally steer a stochastic linear system to a final target Gaussian distribution at the end of a finite interval. In parallel, we formulate the problem to maintain a stationary Gaussian state distribution by constant state feedback for time-invariant dynamics. We also present sufficient conditions of optimality which in the case of finite-horizon take the form of a Schr\"odinger-like system of equations.

{The ability to specify the mean value of the state-vector reduces to the problem discussed at the start of the introduction. More specifically, since $\E\{x(t)\}=:\bar x(t)$ satisfies \eqref{eq:linearsystem}, controllability of $(A,B)$ is necessary and sufficient to specify $\bar x(T)$ at the end of the interval and this is effected by a deterministic mean value for the input process. Likewise, the mean value for a stationary input must satisfy \eqref{eq:steadystate} to attain $\bar x(t)\equiv\xi$ for a stationary state process. Thus, throughout and without loss of generality we assume that all processes have zero-mean and we only focus on our ability to assign the state-covariance in those two instances.}

\subsection{Finite-horizon optimal steering}\label{sec:finitehorizon}
Consider  the controlled evolution
\begin{align}\label{controlled}
dx^u(t)&=A(t)x^u(t)dt+B(t)u(t)dt+B_1(t)dw(t),\\\nonumber &\quad x^u(0)=x_0\mbox{ a.s.}
\end{align}
where $x_0$ an $n$-dimensional Gaussian vector independent of the standard $p$-dimensional Wiener process $\{w(t)\mid 0\le t\le T\}$ and with density
\begin{equation}\label{initial}\rho_0(x)=(2\pi)^{-n/2}\det (\Sigma_0)^{-1/2}\exp\left(-\frac{1}{2}x'\Sigma_0^{-1}x\right).
\end{equation}
Here, $A$, $B$ and $B_1$ are continuous matrix functions of $t$ taking values in $\mR^{n\times n}$, $\mR^{n\times m}$ and $\mR^{n\times p}$, respectively, $\Sigma_0$ is a symmetric positive definite matrix, and $T<\infty$ represents the end point of a time interval of interest.
 Suppose we also have a ``target'' Gaussian end-point distribution
 \begin{equation}\label{final}\rho_T(x)=(2\pi)^{-n/2}\det (\Sigma_T)^{-1/2}\exp\left(-\frac{1}{2}x'\Sigma_T^{-1}x\right),
\end{equation}
where we also assume $\Sigma_T$ symmetric and positive definite.
The uncontrolled evolution $x^{u\equiv 0}=\{x(t) \mid 0\le t\le T\}$ may be thought to represent a ``prior," or reference evolution, for which, in general, $x(T)$ is not distributed according to $\rho_T$. Thus, we seek the least-effort strategy to steer (\ref{controlled}) to the desired final probability density. To this end, let  $\mathcal U$ represent the family of {\em adapted}, {\em finite-energy} control functions such that (\ref{controlled}) has a strong solution and $x^u(T)$ is distributed according to \eqref{final}. Thus, $u\in\mathcal U$ is such that  $u(t)$ only depends on $t$ and on $\{x^u(s)\mid  0\le s\le t\}$ for each $t\in [0,T]$, satisfies
\[
J(u):=\E\left\{\int_0^Tu(t)' u(t) \,dt\right\}<\infty,
\]
and forces $x^u(T)$ to be distributed according to \eqref{final}.
Therefore, $\mathcal U$ represents the class of {\em admissible} control inputs.  The existence of such control inputs will be established in the following section, i.e., that ${\mathcal U}$ is not empty. At present, assuming this to be the case, we formulate the following {\em  Bridge Problem:}

\begin{problem}\label{formalization}  Determine
$
u^*:= \argmin_{u\in \mathcal U} \,J(u)$.
\end{problem}

{We point out that when $BB'\neq B_1B'_1$, no interpretation of this problem as a classical Schr\"{o}dinger bridge \cite{W} via the Girsanov transformation is possible since, in this case, the reference and controlled measures on path spaces are {mutually} singular; this is due to the fact that the martingale part of the two evolutions are different.
In spite of this, precisely the same completion of the squares argument used in \cite[Section II]{chen2014optimal} yields the sufficient conditions in Proposition \ref{prop:riccati} and shows that a {\em control-theoretic view} of the Schr\"odinger bridge problem \cite{dai1991stochastic} carries throught in this more general setting.}

\begin{prop}\label{prop:riccati} Let $\{\Pi(t) \mid 0\le t\le T\}$ be a solution of the matrix Riccati equation
\begin{equation}\label{R1}
\dot{\Pi}(t)=-A(t)'\Pi(t)-\Pi(t)A(t)+\Pi(t)B(t)B(t)'\Pi(t).
\end{equation}
Define the feedback control law
\begin{equation}\label{optcontr}
u(x,t):=-B(t)'\Pi(t)x
\end{equation}
and let $x^{u}=x^*$ be the Gauss-Markov process
\begin{align}\label{optevolution}
dx^*(t)&=\left(A(t)-B(t)B(t)'\Pi(t)\right)x^*(t)dt+B_1(t)dw(t),\\
&\quad \mbox{with }x^*(0)=x_0\mbox{ a.s. }\nonumber
\end{align}
If $x^*(T)$ has probability density $\rho_T$, then $u(x^*(t),t)=u^*(t)$, i.e., it is the solution to Problem \ref{formalization}.
\end{prop}

Now, in contrast to the standard LQR problem where the terminal cost provides a boundary value for the differential Riccati equation, here the boundary value $\Pi(0)$ is unspecified and needs to be selected so as to ensure that \eqref{optcontr} drives the state to the desired final distribution. In \cite{chen2014optimal} we when $B=B_1$, the mapping between $\Pi(T)$ and $\Sigma_T$ is onto with \eqref{R1} having no finite escape-time and, thereby, that steering is always possible. However, it was also noted in \cite{chen2014optimal} that $\Pi(T)$ may be indefinite, placing the analysis outside of standard LQR theory.
Thus, in the present more general case we also need to resort to an approach that departs from classical LQR in order to determine the appropriate solutions of (\ref{R1}). Below we recast Proposition \ref{prop:riccati} in the form of a Schr\"odinger system.

Let $\Sigma(t):=\E\left\{x^*(t)x^*(t)'\right\}$ be the state covariance of \eqref{optevolution} and assume that the conditions of the proposition hold. Then
\begin{align}\nonumber\dot{\Sigma}(t)&= \left(A(t)-B(t)B(t)'\Pi(t)\right)\Sigma(t)\\&\hspace*{-5pt}+\Sigma(t)\left(A(t)-B(t)B(t)'\Pi(t)\right)'+B_1(t)B_1(t)'\label{sigmaevol}
\end{align}
holds together with the two boundary conditions
\begin{equation}\label{BND}
\Sigma(0)=\Sigma_0, \quad \Sigma(T)=\Sigma_T.
\end{equation}
Further, since $\Sigma_0>0$, $\Sigma(t)$ is positive definite on $[0,T]$.
Now define
$${\rm H}(t):=\Sigma(t)^{-1}-\Pi(t).
$$
A direct calculation using (\ref{sigmaevol}) and (\ref{R1}) leads to \eqref{Schr2} below.
We have therefore derived a {\em nonlinear} Schr\"{o}dinger system
\begin{subequations}\label{Schr1234}
\begin{eqnarray}\label{Schr1}
\hspace*{-5pt}\dot{\Pi} &=&-A '\Pi -\Pi A +\Pi B B '\Pi \\
\hspace*{-55pt}\dot{\rm H} &=&-A '{\rm H} -{\rm H} A -{\rm H} B B '{\rm H} \label{Schr2}\\
&&\hspace*{1cm}+\left(\Pi +{\rm H} \right)\left(B B '-B_1 B_1 '\right)\left(\Pi +{\rm H} \right).\nonumber\\
\hspace*{-5pt}\Sigma_0^{-1}&=&\Pi(0)+{\rm H}(0)\label{Schr3}\\
\hspace*{-5pt}\Sigma_T^{-1}&=&\Pi(T)+{\rm H}(T).\label{Schr4}
\end{eqnarray}
\end{subequations}
Indeed, in contrast to the case when $B=B_1$ (see \cite{chen2014optimal}),  the two Riccati equations in (\ref{Schr1234}) are coupled
not only through their boundary values (\ref{Schr3}-\ref{Schr4}) but also in a nonlinear manner through their dynamics in (\ref{Schr2}).
Clearly, the case $\Pi(t)\equiv 0$ corresponds to the situation where the uncontrolled evolution already satisfies the boundary marginals and, in that case, ${\rm H}(t)^{-1}$ is simply the prior state covariance.
We summarize our conclusion in the following proposition.
\begin{prop} {Assume that} $\{(\Pi(t),{\rm H}(t)) \mid 0\le t\le T\}$ satisfy (\ref{Schr1})-(\ref{Schr4}). Then the feedback control law (\ref{optcontr}) is the solution to Problem \ref{formalization} and the corresponding optimal evolution is given by (\ref{optevolution}).
\end{prop}

{The existence and uniqueness of solutions for the Schr\"{o}dinger system  is quite challenging already in the classical case where the two dynamical equations are uncoupled and where major contributions are due to Fortet \cite{fortet}, Beurling \cite{Beurling}, Jamison \cite{Jamison}, F\"{o}llmer \cite{W}, see also  \cite{GP,chen2014optimal}. It is therefore hardly surprising that at present we don't know how to prove existence of solutions for (\ref{Schr1})-(\ref{Schr4})
\footnote{A numerical scheme based on successive approximations appears to be unstable and does not produce a fix point in general. In this case, such a scheme could consist of solving \eqref{Schr1} backwards in time starting from $\Pi(T)$, computing initial conditions for \eqref{Schr2} using \eqref{Schr3}, solving \eqref{Schr2}  forward in time to compute $H(T)$ so as to update $\Pi(T)$ using \eqref{Schr4} and repeating the cycle. A similar idea was carried out by Fortet \cite{fortet} in the classical setting, whereas a more powerful technique based on the Hilbert metric was explored recently in \cite{GP} for a Schr\"odinger system on finite spaces.}
A direct proof of existence of solutions for (\ref{Schr1})-(\ref{Schr4}) would in particular imply {\em feasibility} of Problem \ref{formalization}, i.e., that $\mathcal U$ is nonempty and that there exists a minimizer.
At present we do not have a proof that a minimizer exists. However, in Section \ref{sec:feasibility1} we establish that the set of admissible controls $\mathcal U$ is not empty and in Section \ref{sec:optimalsteering} we provide an approach that allows constructing suboptimal controls incurring cost that is arbitrarily close to $\inf_{u\in{\mathcal U}}J(u)$.}

{
\subsection{Infinite-horizon optimal steering}
Suppose now that $A$, $B$ and $B_1$ do not depend on time and that the pair $(A,B)$ is controllable.
We seek a constant state feedback law $u(t)=-Kx(t)$
to maintain a stationary state-covariance $\Sigma>0$ for \eqref{controlled}. In particular, we are interested
in one that minimizes the expected input power (energy rate)
\begin{eqnarray}\label{eq:power}
J_{\rm power}(u)&:=&\E\{u'u\}
\end{eqnarray}
and thus we are led to the following problem\footnote{ An equivalent problem is to minimize
$\lim_{T\to \infty}\frac{1}{T}\E\left\{\int_0^Tu(t)'u(t)dt\right\}$
for a given terminal state covariance as $T\to\infty$.}.
\begin{problem}\label{problem2} Determine $u^*$ that minimizes $J_{\rm power}(u)$ over all $u(t)=-Kx(t)$ such that 
\begin{equation}\label{feedbackdynamics}
dx(t)=(A-BK)x(t)dt+B_1dw(t)
\end{equation}
admits
\begin{equation}\label{invdensity} 
\rho(x)=(2\pi)^{-n/2}\det (\Sigma)^{-1/2}\exp\left(-\frac{1}{2}x'\Sigma^{-1}x\right)
\end{equation}
as invariant probability density.
\end{problem}

{Interestingly, the above problem may not have a solution in general since not all values for $\Sigma$ can be maintained by state feedback. {In fact, Theorem \ref{admissiblestate3} in Section \ref{sec:feasibility2}, provides conditions that ensure $\Sigma$ is admissible
as a stationary state covariance for a suitable input.
Moreover, as it will be apparent from what follows, even when the problem is feasible, i.e., there exist controls which maintain $\Sigma$, an optimal control may fail to exist}.

Let us start by observing that the problem admits the following finite-dimensional reformulation. Let $\mathcal K$ be the set of all $m\times n$ matrices $K$ such that the corresponding feedback matrix $A-BK$ is Hurwitz. Observe that
\[
\E\{u'u\}=\E\{x'K'Kx\}=\tr(K\Sigma K')
\]
Then Problem \ref{problem2} reduces to finding a $m\times n$ matrix $K^*\in\mathcal K$ which minimizes the criterion
\begin{equation}\label{criterion}
J(K)=\tr\left(K\Sigma K'\right)
\end{equation}
subject to the constraint
\begin{equation}
(A-BK)\Sigma+\Sigma(A'-K'B')+B_1B_1'=0.\label{constraint}
\end{equation}
Now, consider the Lagrangian function
 \begin{eqnarray}
 \mathcal{L}(K,\Pi)&=&\tr\left(K\Sigma K'\right)\\\nonumber&&\hspace*{-1cm}+\tr\left(\Pi((A-BK)\Sigma+\Sigma(A'-K'B')+B_1B_1')\right)
 \end{eqnarray}
 which is a simple quadratic form in the unknown $K$.
Observe that $\mathcal K$ is {\em open}, {hence a minimum point may fail to exist. Nevertheless,} at any point $K\in\mathcal K$ we can take a directional derivative in any direction $\delta K\in\mR^{m\times n}$ to obtain
$$
\delta \mathcal{L}(K,\Pi;\delta K)=\tr\left(\left(\Sigma K'+K\Sigma-\Sigma\Pi B-B'\Pi\Sigma\right)\delta K\right). 
$$
Setting $\delta \mathcal{L}(K,\Pi;\delta K)=0$ for all variations, {which is a sufficient condition for optimality, we get the form 
\begin{equation}\label{optgain}
K^*=B'\Pi.
\end{equation}
To compute $K^*$, we calculate the multiplier $\Pi$ as a maximum point of the dual functional
\begin{eqnarray}\label{dual}
G(\Pi)&=& \mathcal{L}(K^*,\Pi)\\\nonumber
&=&\tr\left(\left(A'\Pi+\Pi A-\Pi BB'\Pi\right)\Sigma+\Pi B_1B_1'\right).
\end{eqnarray}
The unconstrained maximization of the concave functional $G$ over symmetric $n\times n$ matrices produces matrices $\Pi^*$ which satisfy
(\ref{constraint}), namely
\begin{equation}\label{constraint'}
(A-BB'\Pi^*)\Sigma+\Sigma(A'-\Pi^*BB')+B_1B_1'=0.
\end{equation} 
There is no guarantee, however, that $K^*=B'\Pi^*$ is in $\mathcal K$, namely that $A-BB'\Pi^*$ is Hurwitz. Nevertheless, since (\ref{constraint'}) is satisfied, the spectrum of $A-BB'\Pi^*$ lies in the {\em closed} left half-plane. }
{ Thus, our variational analysis leads to the following result.}
\begin{prop}\label{prop:prop1}
Assume that there exists a symmetric matrix $\Pi$ such that $A-BB'\Pi$ is a Hurwitz matrix and
\begin{equation}\label{sigmastat}
(A-BB'\Pi)\Sigma+\Sigma(A-BB'\Pi)'+B_1B'_1=0
\end{equation}
holds.
Then 
\begin{equation}\label{statoptcontr}
u^*(t)=-B'\Pi x(t)
\end{equation}
is the solution to Problem \ref{problem2}. 
\end{prop}

{ We now draw a connection to some classical results due to  Jan Willems \cite{willems1971least}.  In our setting, minimizing \eqref{eq:power} is equivalent to minimizing
\begin{equation}\label{eq:alternate}
J_{\rm power}(u)+\E\{x'Qx\}
\end{equation}
for an arbitrary symmetric matrix $Q$ since the portion
\[
\E\{x'Qx\}=\tr\{Q\Sigma\}
\]
is independent of the choice of $K$.
On the other hand, minimization of \eqref{eq:alternate} for specific $Q$, but without the constraint that $\E\{xx'\}=\Sigma$, was studied by Willems \cite{willems1971least} and is intimately related to the {\em maximal} solution of the Algebraic Riccati Equation (ARE)
\begin{equation}\label{ARE}
A'\Pi+\Pi A-\Pi BB'\Pi+Q=0.
\end{equation}
Under the assumption that the Hamiltonian matrix
$$H=\left[\begin{matrix} A &-BB'\\-Q & -A'\end{matrix}\right]
$$
has no pure imaginary eigenvalues, Willems' result states  that $A-BB'\Pi$ is Hurwitz and that \eqref{statoptcontr} is the optimal solution.

Thus, starting from a symmetric matrix $\Pi$ as in Proposition \ref{prop:prop1}, we can define $Q$ using
\[
Q=-A'\Pi-\Pi A+\Pi BB'\Pi.
\]
Since by Willems' results, \eqref{ARE} has at most one ``stabilizing'' solution $\Pi$, the matrix in the proposition coincides with the maximal solution to \eqref{ARE}. Therefore, if our original problem has a solution, this same solution can be recovered by solving for the maximal solution of a corresponding ARE, for a particular choice of $Q$.
Interestingly, neither $\Pi$ nor $Q$, corresponding to an optimal control law for which \eqref{sigmastat} holds, are unique, whereas $K$ is. The computation and the uniqueness of the optimal gain $K$ will be discussed later on in Section \ref{sec:mestationary}.
}

\section{Controllability of state statistics}\label{sec:variational}

We now return to the ``controllability'' question
of whether there exist admissible control
to steer the controlled evolution
\begin{align}\label{eq:stochasticsystem}
dx(t)=&Ax(t)dt+Bu(t)dt + B_1dw(t)\\\nonumber &\mbox{with } x(0)=x_0 \mbox{ a.s. }
\end{align}
to a target Gaussian distribution at the end of a
finite interval $[0,\,T]$, or, for the stationary case, whether a stationary Gaussian distribution can be achieved by constant state feedback.
From now on, we assume that
$A\in\mR^{n\times n}$, $B\in\mR^{n\times m}$ and $B_1\in\mR^{n\times p}$, are time-invariant and that $(A,B)$ is controllable.
In view of the earlier analysis, we search over controls that are linear functions of the state, i.e.,
\begin{equation}\label{eq:feedback}
u(t)=-K(t) x(t), \;\mbox{ for }t\in[0,T],
\end{equation}
and where $K$ is constant and $A-BK$ Hurwitz for the stationary case.

\subsection{Finite-interval steering by state-feedback}\label{sec:feasibility1}

We assume that $\E\{x_0\}=0$ while $\E\{x_0x_0'\}=\Sigma_0$. The state covariance
\[\Sigma(t):=\E\{x(t)x(t)'\}
\]
of \eqref{controlled}
with input as in \eqref{eq:feedback} satisfies the Lyapunov differential equation
\begin{equation}\label{eq:covariancedynamics}
\dot\Sigma(t)=(A-BK(t))\Sigma(t)+\Sigma(t)(A-BK(t))'+B_1B_1'
\end{equation}
and $\Sigma(0)=\Sigma_0$. Regardless of the choice of
$K(t)$, \eqref{eq:covariancedynamics} specifies dynamics that leave the cone of positive semi-definite symmetric matrices
\[
\cS_n^+:=\{\Sigma \mid \Sigma\in\mR^{n\times n},\;\Sigma=\Sigma'\geq 0\}
\]
invariant. To see this, note that the solution to \eqref{eq:covariancedynamics} is of the form
\[
\Sigma(t)=\hat\Phi(t,0)\Sigma_0\hat\Phi(t,0)' +
\int_0^t\hat\Phi(t,\tau)B_1B_1'\hat\Phi(t,\tau)' d\tau
\]
where $\hat\Phi(t,0)$ satisfies
\[
\frac{\partial \hat\Phi(t,0)}{\partial t}=(A-BK(t))\hat\Phi(t,0)
\]
and $\hat\Phi(0,0)=I$, the identity matrix; i.e., $\hat\Phi(t,0)$ is the state-transition matrix of the system $\dot x(t)=(A-BK(t))x(t)$.

Assuming $\Sigma_0>0$, it follows that $\Sigma(t)>0$ for all $t$ and finite $K(\cdot)$. Our interest is in our ability to specify $\Sigma(T)$ via a suitable choice of $K(t)$. To this end, we define
\[
U(t):=-\Sigma(t)K(t)',
\]
we observe that $U(t)$ and $K(t)$ are in bijective correspondence provided that $\Sigma(t)>0$, and we now consider the differential Lyapunov system
\begin{equation}\label{eq:diffLyapunov}
\dot\Sigma(t)=A\Sigma(t)+\Sigma(t)A'+BU(t)'+U(t)B'.
\end{equation}
Reachability/controllability of a differential system such as \eqref{eq:linearsystem}, or \eqref{eq:diffLyapunov}, is the property that with suitable bounded control input $u(t)$, or $U(t)$, respectively, the solution can be driven to any finite value. Interestingly, if any of \eqref{eq:linearsystem} and \eqref{eq:diffLyapunov} is controllable, so is the other. But, more importantly, when \eqref{eq:diffLyapunov} is controllable, the control authority allowed is such that steering from one value for the covariance to another can be done by remaining within the non-negative cone. This is stated as our first theorem below.

\begin{thm}\label{thm:thm1}
The Lyapunov system \eqref{eq:diffLyapunov} is controllable iff $(A,B)$ is a controllable pair. Furthermore, if \eqref{eq:diffLyapunov} is controllable, given any two positive definite matrices $\Sigma_0$ and $\Sigma_T$ and an arbitrary $Q\geq 0$, there is a smooth input $U(\cdot)$ so that the solution of
the (forced) differential equation
\begin{equation}\label{eq:diffLyapunov2}
\dot\Sigma(t)=A\Sigma(t)+\Sigma(t)A'+BU(t)'+U(t)B' +Q
\end{equation}
satisfies the boundary conditions $\Sigma(0)=\Sigma_0$ and $\Sigma(T)=\Sigma_T$
and $\Sigma(t)>0$ for all $t\in[0,T]$.
\end{thm}

\begin{proof}
We first establish equivalence of the controllability of \eqref{eq:linearsystem} and \eqref{eq:diffLyapunov}. Define
$
S(t):=e^{-At}\Sigma(t)e^{-A't}.
$
In these new ``coordinates'' \eqref{eq:diffLyapunov} becomes
\[
\dot S(t)=e^{-At}BU(t)'e^{-A't}+e^{-At}U(t)B'e^{-A't},
\]
and upon re-naming $V(t)=e^{-At}U(t)$ as the input,
\begin{equation}\label{eq:complete}
\dot S(t)=e^{-At}BV(t)'+V(t)B'e^{-A't}.
\end{equation}
Assuming that $(A,B)$ is a controllably pair, the system
\begin{equation}\label{eq:half}
\dot X(t)=e^{-At}BV(t)',
\end{equation}
where each column of $V(t)'$ serves as input that drives the corresponding column of $X(t)$ is clearly
controllable since the controllability grammian
\[
G(T):=\int_0^Te^{-A\tau}BB'e^{-A'\tau}d\tau
\]
is invertible. Thus, by a suitable choice of $V(t)$ we can drive  \eqref{eq:half} to any final state $X(T)$ and, thus, we can drive \eqref{eq:complete} to any final state $S(T)=X(T)+X(T)'$.

The converse is straightforward. If $(A,B)$ is not controllable, then there is a matrix $C$ such that $Ce^{-At}B=0$. It follows that $C\dot S(t)C'=0$ and therefore $S(t)$ remains invariant when restricted to a certain subspace.

We now want to establish that there is a control input $U(t)$ so that the solution to \eqref{eq:diffLyapunov2} remains within the positive cone and satisfies the boundary conditions. We show this, and in fact, a stronger argument for a special case where $A$ is a shift matrix and $B$ is vectorial, and then explain why the general case can be reduced to this one.

So, we now establish that there is a smooth (infinitely differentiable) control input $U(t)$ so that $\Sigma(t)$ remains within the positive cone and satisfies the boundary conditions. We further claim (and show below) that such a control can always be chosen to satisfy arbitrary starting and ending boundary conditions $U(0)$ and $U(T)$ of its own. We show this for the special case where $A$ is a shift matrix of size $k$.  For specificity in the steps of the proof, we subscribe the size of matrices in the notation
\begin{equation}\label{eq:shift}
A_k:=\left[\begin{matrix} 0_{k-1} & I_{k-1}\\ 0 & 0_{k-1}^\prime\end{matrix}\right]\mbox{ and }
B_k:=\left[\begin{matrix} 0_{k-1} \\1\end{matrix}\right].
\end{equation}
Here also, $I_k$ denotes the identity matrix of size $k$, and $0_k$ the column vector of size $k$ that has all entries zero.
We will show by induction on $k$ that, for any $k\times k$ matrix $Q_k\geq 0$, the system
\begin{equation}\label{eq:kth}
\dot\Sigma(t)=A_k\Sigma(t)+\Sigma(t)A_k'+B_kU_k(t)'+U_k(t)B_k' +Q_k
\end{equation}
can be steered between positive-definite boundary values while $\Sigma(t)$, which is now $k\times k$, remains positive-definite and the control satisfies arbitrary starting and ending values. The statement is true for $k=1$. In this case, the system is in the form
\begin{equation}\label{eq:U1}
\dot\Sigma(t)=2U_1(t)+Q
\end{equation}
with all entries scalar. Positivity of $\Sigma(t)$ dictates that
\[
\Sigma_0+ 2\int_0^t U_1(\tau)d\tau +Qt>0 \mbox{ for all }t,
\]
while the boundary conditions dictate that
\[
\Sigma_0+ 2\int_0^T U_1(\tau)d\tau +QT=\Sigma_T.
\]
Clearly, these can be met along with any boundary conditions on $U_1(t)$ along with the smoothness requirement.
An example of such an interpolating function is $\Sigma(t)=e^{h(t)}>0$ where
    \begin{eqnarray*}
        h(t)&=&a_0+b_0 t+\frac{a_T-a_0-Tb_0}{T^2}t^2\\
        &&+\frac{Tb_0+Tb_T-2a_T+2a_0}{T^3}t^2(t-T)
    \end{eqnarray*}
and
    \begin{eqnarray*}
        a_0 &=& \log(\Sigma_0) \\
        a_T &=& \log(\Sigma_T) \\
        b_0 &=& (2U_1(0)+Q)/{\Sigma_0}\\
        b_T &=& (2U_1(T)+Q)/{\Sigma_T}.
    \end{eqnarray*}
The polynomial $h(t)$ is in fact a Hermite polynomial satisfying
    \begin{eqnarray*}
    h(0) &=& a_0,~~h(T)=a_T\\
    \dot{h}(0) &=& b_0,~~\dot{h}(T)=b_T.
    \end{eqnarray*}
It is easy to see that $\Sigma(t)=e^{h(t)}$ satisfies
    \begin{eqnarray*}
        \Sigma(0)&=&\Sigma_0,~~\Sigma(T)=\Sigma_T\\
        \dot{\Sigma}(0) &=& 2U_1(0)+Q\\
        \dot{\Sigma}(T) &=& 2U_1(T)+Q,
    \end{eqnarray*}
and $U_1(t)$ can be computed from \eqref{eq:U1}.

We now assume that the claim is valid for $k=n-1$ and argue that it is also true for $k=n$.
Before we do so we note that, for any size of matrices, \eqref{eq:diffLyapunov2}
implies that
\begin{equation}\label{eq:kth_b}
\Pi\dot\Sigma(t)\Pi=\Pi A\Sigma(t)\Pi+\Pi \Sigma(t)A'\Pi +\Pi Q\Pi,
\end{equation}
where $\Pi:=\Pi_{{\mathcal R}(B)^\perp}$ is the projection onto the orthogonal compelement of the range of $B$, since $\Pi B=0$. Conversely, if \eqref{eq:kth_b} holds, there exists a $U(t)$ so that \eqref{eq:diffLyapunov2} holds. To see this, note that
the map
\begin{equation}\label{eq:g}
\g_B\;:\;\cS_n\to\cS_n\;:\; Y\mapsto \Pi_{{\mathcal R}(B)^\perp}Y\Pi_{{\mathcal R}(B)^\perp}
\end{equation}
is self-adjoint. Throughout, $\mathcal{S}_n$ denotes the linear vector space of symmetric matrices of dimension $n$,
\[
\Pi_{\cR(B)^\perp}:=I-B(B'B)^{-1}B'
\]
denotes the projection onto the orthogonal complement of the range of $B$ (when $B$ is singular, the inverse needs to be replace by a pseudoinverse), and
where $I$ denotes identity matrix.
Since $\g_B$ is self-adjoint, the orthogonal complement of its range is precisely its null space, which according to the lemma in Appendix \ref{sec:appendix}, is also the range of
\begin{equation}\label{eq:f}
\f_B\;:\;\mR^{n\times m}\to \mathcal{S}_n\;:\;
X\mapsto BX'+XB'.
\end{equation}
But
\[
\dot\Sigma(t)-(A\Sigma(t)+\Sigma(t)A'+Q)
\]
when projected onto the range of $\g_B$ is identically zero (since \eqref{eq:kth_b} holds). Hence, \eqref{eq:diffLyapunov2} also holds for a suitable $U(t)$. (In other words, the extra directions that \eqref{eq:kth_b} does not already restrict can be freely adjusted by a proper choice of $U(t)$ since they are in the range of $\f_B$.) The fact that we can always select $U(t)$ to be smooth, provided of course that $\Sigma(t)$ is smooth, follows since $\g_B$ is linear. Also, similarly as in the $k=1$ case, we can select $U(t)$ to satisfy arbitrary boundary conditions $U(0)$ and $U(T)$ of its own.

Let us now return to the induction argument. Equation \eqref{eq:kth} for $k=n$, is equivalent to
\begin{equation}\label{eq:kth2c}
\Pi_n\dot\Sigma(t)\Pi_n=\Pi_n A_n\Sigma(t)\Pi_n+\Pi_n \Sigma(t)A_n'\Pi_n +\Pi_n Q_n\Pi_n
\end{equation}
where
\[
\Pi_n =\left[\begin{matrix} I_{n-1} &0_{n-1}\\0_{n-1}^\prime & 0\end{matrix}\right].
\]
If we partition
\[
\Sigma(t)=\left[\begin{matrix}\Sigma_1(t) &\sigma_2(t)\\\sigma_2(t)^\prime & \sigma_3(t)\end{matrix}\right]
\]
where $\Sigma_1$ is $(n-1)\times(n-1)$, $\sigma_2$ is a column vector, and $\sigma_3$ a scalar, then
\eqref{eq:kth2c} becomes
\begin{align}\nonumber
\left[\begin{matrix}\dot\Sigma_1(t) &0_{n-1}\\0_{n-1}^\prime &0\end{matrix}\right]
&=M
\left[\begin{matrix}\Sigma_1(t) &0_{n-1}\\\sigma_2(t)^\prime & 0\end{matrix}\right]+ \left[\begin{matrix}\Sigma_1(t) &\sigma_2(t)\\0_{n-1}^\prime & 0\end{matrix}\right]M^\prime\\ & \phantom{ = }\;+
\left[\begin{matrix}Q_1 &0_{n-1}\\0_{n-1}^\prime &0\end{matrix}\right]\label{eq:reduced}
\end{align}
where $Q_1$ is the $(n-1)\times(n-1)$ block of $Q$ and
\begin{align*}
M&=\Pi A_n\\
&=\left[\begin{matrix} 0_{n-1}& I_{n-1} &\\0& 0_{n-1}^\prime \end{matrix}\right]\\
&=\left[\begin{matrix} A_{n-1} &B_{n-1}\\0_{n-1}^\prime&0 \end{matrix}\right]
\end{align*}
after we group its entries consistent with the partition of $\Sigma$. But now, \eqref{eq:reduced} is in the form
\[
\dot\Sigma_1(t)=A_{n-1}\Sigma_1(t)+\Sigma_1(t)A_{n-1}^\prime +
B_1\sigma_2(t)^\prime +\sigma_2(t)B_1^\prime+Q_1.
\]
Since the matrices in this one are of size $(n-1)\times(n-1)$, by our hypothesis, we can find a control $U(t)$ which will then identify with $\sigma_2(t)$. The boundary conditions for $U(t)$ are dictated by the boundary conditions for $\Sigma(t)$. The final entry of $\Sigma(t)$, $\sigma_3(t)$ is not restricted in any way other than being in agreement with the boundary conditions of $\Sigma$. The values are the two ends, $\sigma_3(0)$ and $\sigma_3(T)$ are admissible since $\Sigma_0>0$ as well as $\Sigma_T>0$. Thus, we can choose a smooth function for $\sigma_3(t)$ that takes values large enough in $(0,T)$ so that $\Sigma(t)>0$ throughout.

A final point is needed to complete the proof. For an arbitrary controllable pair $(A,B)$
it is well known that there exists a constant $K$ and a vector $v$ such that $(A-BK,Bv)$ is controllable (Heymann's lemma, see \cite{hautus1977simple}). Further, $K$ can be chosen so that $A-BK$ has all eigenvalues at the origin, hence it is equivalent to a shift matrix. Thus, we can choose $K$ and $v$ such that,
after a similarity transformation, $(A-BK,Bv)$ becomes $(A_n,B_n)$ (in the notation of \eqref{eq:shift}). The statement of the theorem is invariant to similarity transformation as well as to action of the feedback group $A\mapsto A-BK$. Further, replacing $B$ with $Bv$ corresponds to selecting a portion of the allowed control authority, and we have already shown the theorem for this case which is more stringent. This completes the proof.
\end{proof}

\subsubsection*{Finite-interval steering via external input}\label{sec:finite_external}
It is interesting to observe an equivalence between steering the state-covariance of a stochastic system by state-feedback, and modeling changes in the state covariance as due to an external input process for the case where $B=B_1$.
Specifically, given the Gauss-Markov model
\[
dx(t)=Ax(t)dt+Bdy(t)
\]
and a path for the evolution of its state-covariance $\{\Sigma(t) \mid t\in[0,T]\}$ that satisfies \eqref{eq:diffLyapunov2} for some $U(t)$, we are interested in a possible external input process $y(t)$ that is responsible for steering the state covariance through the specified path. That is, we want to model the state-evolution by postulating a suitable process $y(t)$.
We observe that the Gauss-Markov process
\begin{eqnarray}\label{eq:filter0}
 d\xi(t)&=&(A-BK(t))\xi(t)dt+Bdw(t)\\\nonumber
 dy(t)&=&-K(t)\xi(t)dt+dw(t),
\end{eqnarray}
with $\E\{\xi(0)\xi(0)'\}=\Sigma_0$
and
\[
K(t)=-U(t)'\Sigma(t)^{-1}.
\]
It follows that
\[
d\xi(t)=A\xi(t)dt+Bdy(t)
\]
and therefore $\xi(t)$ and $x(t)$ share the same statistics. On the other hand,
the state covariance of \eqref{eq:filter0} satisfies \eqref{eq:diffLyapunov2}.

\subsection{Assignability of stationary state covariances via state-feedback}\label{sec:feasibility2}

We are interested to steer and maintain the system through static state-feedback
 \begin{equation}\label{eq:feedback2}
 u(t)=-Kx(t)
 \end{equation}
at an equilibrium distribution with a given state-covariance $\Sigma$. Due to linearity, the distribution will then be Gaussian. It is clear that depending on the value of $\Sigma$ this may not always possible. The family of admissible stationary state-covariances are given below.

Assuming that $A-BK$ is a Hurwitz matrix, which is necessary for
the state process $\{x(t)\mid t\in[0,\infty)\}$ to be stationary, the (stationary) state-covariance $\Sigma=\E\{x(t)x(t)'\}$ satisfies the algebraic Lyapunov equation
\begin{equation}\label{eq:lyapunov2}
(A-BK)\Sigma + \Sigma (A-BK)'=-B_1B_1'.
\end{equation}
Thus, the equation
\begin{subequations}\label{eq:equivalent}
\begin{eqnarray}\label{eq:lyapunov3}
&&A\Sigma + \Sigma A'+B_1B_1'+BX'+XB'=0\\
&&\mbox{can be solved for $X$},\nonumber
\end{eqnarray}
which in particular can be taken
to be $X=-\Sigma K'$. The solvability of \eqref{eq:lyapunov3} is obviously a necessary condition for $\Sigma$ to qualify as a stationary state-covariance attained via feedback. Alternatively, \eqref{eq:lyapunov3} is equivalent to saying that
\begin{equation}\label{eq:range}
A\Sigma + \Sigma A'+B_1B_1'\in {\mathcal R}(\f_B).
\end{equation}
The latter can be expressed as a rank condition \cite[Proposition 1]{georgiou2002structure} in the form
\begin{equation}\label{eq:rank2}
{\rm rank}\left[\begin{matrix}
A\Sigma+\Sigma A'+B_1B_1' & B\\
B & 0
\end{matrix}\right]
=
{\rm rank}\left[\begin{matrix}
0 & B\\
B & 0
\end{matrix}\right].
\end{equation}
Also, in view of Lemma \ref{lemma:space}, \eqref{eq:range} is equivalent to
\begin{equation}\label{eq:null}
A\Sigma + \Sigma A'+B_1B_1'\in {\mathcal N}(\g_B).
\end{equation}
\end{subequations}
Therefore, the conditions (\ref{eq:lyapunov2}-\ref{eq:null}), which are all equivalent, are necessary for the existence of a state-feedback gain $K$ that ensures $\Sigma>0$ to be the stationary state covariance of \eqref{controlled}.

{ Conversely, given $\Sigma>0$ that satisfies \eqref{eq:equivalent} and $X$ the solution to \eqref{eq:lyapunov3},
then \eqref{eq:lyapunov2} holds with $K=-X'\Sigma^{-1}$. Provided $A-BK$ is a Hurwitz matrix, $\Sigma$ is admissible stationary covariance. The property of $A-BK$ being Hurwitz can be guaranteed when
$(A-BK,\,B_1)$ is a controllable pair. In turn, controllability of $(A-BK,\,B_1)$ is guaranteed when $\cR(B)\subseteq \cR(B_1)$. Thus, we have established the following.}

\begin{thm}\label{admissiblestate3}
Consider the Gauss-Markov model \eqref{controlled}
and assume that $\cR(B)\subseteq \cR(B_1)$. A positive-definite matrix $\Sigma$ can be assigned as the stationary state covariance via a suitable choice of state-feedback if and only if $\Sigma$ satisfies any of the equivalent statements (\ref{eq:lyapunov3}-\ref{eq:null}). \end{thm}

Interest in \eqref{eq:null} was raised in \cite{hotz1987covariance} where it was shown to characterize state-covariances that can be maintained by state-feedback.
On the other hand, conditions (\ref{eq:lyapunov3}-\ref{eq:rank2}) were obtained in \cite{georgiou2002structure,georgiou2002spectral}, for
the special case when $B=B_1$, as being necessary and sufficient for a positive-definite matrix to materialize as the state covariance of the system driven by a stationary stochastic process (not-necessarily white).  It should be noted that in \cite{georgiou2002structure}, the state matrix $A$ was assumed to be already Hurwitz so as to ensure stationarity of the state process.
However, if the input is generated via feedback as above, $A$ does not need to be Hurwitz whereas, only $A-BK$ needs to be.

\subsubsection*{Assignability via external input}

We now turn to the question of which positive definite matrices materialize as state covariances of the Gauss-Markov model
\begin{equation}\label{eq:short}
dx(t)=Ax(t)+Bdy(t),
\end{equation}
with $(A,B)$ controllable and $A$ Hurwitz, when driven by some stationary stochastic process $y(t)$. The characterization of admissible state covariances was obtained in \cite{georgiou2002structure} and amounts to the condition that
\[A\Sigma+\Sigma A'\in \cR(\f_B)\]
which coincides with the condition that $\Sigma$ can be assigned as in Theorem~\ref{admissiblestate3} by state-feedback. As in Section \ref{sec:finite_external},
a feedback system can be implemented, separate from \eqref{eq:short}, to generate a suitable input processes to give rise to $\Sigma$ as the state covariance of \eqref{eq:short}.
Specifically,
let $X$ be a solution of
\begin{equation}\label{eq:X}
A\Sigma+\Sigma A'+BX'+XB'=0,
\end{equation}
and
\begin{eqnarray}\label{eq:filter}\nonumber
 d\xi(t)&=&(A-BK)\xi(t)dt+Bdw(t)\\\nonumber
 dy(t)&=&-K\xi(t)dt+dw(t)
\end{eqnarray}
with
\begin{equation}\label{eq:K}
K=\frac12 B'\Sigma^{-1} - X'\Sigma^{-1}.
\end{equation}
Trivially,
\[
d\xi(t)=A\xi(t)dt+Bdy(t),
\]
and therefore, $\xi(t)$ shares the same stationary statistics with $x(t)$. But if $S=\E\{\xi(t)\xi(t)'\}$,
\[
(A-BK)S+S(A-BK)'+BB'=0,
\]
which, in view of (\ref{eq:X}-\ref{eq:K}), is satisfied by $S=\Sigma$.

{
\section{Numerical computation of optimal control}\label{sec:optimalsteering}

Having established feasibility for the problem to steer the state-covariance to a given value at the end of an interval, it is of interest to design efficient methods to compute the optimal controls of Section \ref{sec:variational2}. As an alternative to solving the generalized Schr\"{o}dinger system (\ref{Schr1}-\ref{Schr4}), we formulate the optimization as a semidefinite program in Section \ref{sec:num}, and likewise for the infinite-horizon problem in Section \ref{sec:mestationary}.}

\subsection{Finite interval minimum energy steering of state statistics}\label{sec:num}
We are interested in computing an optimal choice for feedback gain $K(t)$ so that
the control signal $u(t)=-K(t)x(t)$ steers \eqref{controlled}
from an initial state-covariance $\Sigma_0$ at $t=0$ to the final $\Sigma_T$ at $t=T$. The expected control-energy functional
    \begin{eqnarray}\label{eq:functional}
       J(u)&:=& \E\left\{\int_0^T u(t)'u(t)dt\right\}\\
       &=&\int_0^T \tr(K(t)\Sigma(t)K(t)')dt\nonumber
    \end{eqnarray}
needs to be optimized over $K(t)$ so that \eqref{eq:covariancedynamics} holds as well as the boundary conditions
\begin{subequations}\begin{equation}\label{eq:boundary}
\Sigma(0)=\Sigma_0, \mbox{ and }\Sigma(T)=\Sigma_T.
\end{equation}

If instead we sought to optimize over $U(t):=-\Sigma(t)K(t)'$ and $\Sigma(t)$, the functional \eqref{eq:functional} becomes
\[
J=\int_0^T \tr(U(t)'\Sigma(t)^{-1}U(t))dt
\]
which is jointly convex in $U(t)$ and $\Sigma(t)$, while \eqref{eq:covariancedynamics} is replaced by
\begin{equation}\label{eq:diffeqB1U}
\dot\Sigma(t)=A\Sigma(t)+\Sigma(t) A'+BU(t)'+U(t)B'+B_1B_1'
\end{equation}
which is now linear in both. Thus, finally, the optimization can be written as a semi-definite program to minimize
\begin{equation}\label{eq:sdp}
 \int_0^T \tr(Y(t))dt
 \end{equation}
 subject to (\ref{eq:boundary}-\ref{eq:diffeqB1U}) and
 \begin{equation}
\left[\begin{matrix}Y(t)& U(t)' \\U(t) & \Sigma(t)\end{matrix}\right]\ge 0.
\end{equation}
\end{subequations}
This can be solved numerically after discretization in time and a corresponding (suboptimal) gain
recovered as $K(t)=-U(t)'\Sigma(t)^{-1}$.

\subsection{Minimum energy control to maintain stationary state statistics}\label{sec:mestationary}

As noted earlier, a positive definite matrix $\Sigma$ is admissible as a stationary state-covariance provided \eqref{eq:lyapunov3} holds for some $X$ and $A+BX'\Sigma^{-1}$ is a Hurwitz matrix. The condition $\cR(B)\subseteq \cR(B_1)$ is a sufficient condition for the latter to be true always, but it may be true even if $\cR(B)\subseteq \cR(B_1)$ fails (see the example in Section \ref{sec:example}). Either way, the expected input power (energy rate)
\begin{eqnarray}
\E\{u'u\}
&=&\tr(K\Sigma K')
\\\nonumber
&=&\tr(X'\Sigma^{-1}X)
\end{eqnarray}
{ in either
in $K$, or $X$.}
Thus, assuming that $\cR(B)\subseteq \cR(B_1)$ holds, and in case \eqref{eq:lyapunov3} has multiple solutions, the optimal constant feedback gain $K$ can be obtained by solving the convex optimization problem
\begin{equation}\label{eq:XSX}
\min\left\{\tr(K\Sigma K')\mid \mbox{ \eqref{eq:lyapunov3} holds } \right\}.
\end{equation}

\begin{remark}
In case $\cR(B)\not\subseteq \cR(B_1)$, the condition that $A-BK$ be Hurwitz needs to be verified separately. If this fails, we cannot guarantee that $\Sigma$ is an admissible stationary state-covariance that can be maintained with constant state-feedback. However, it is always possible to maintain a state-covariance that is arbitrarily close. To see this, consider the control
\[
K_\epsilon = K + \frac12 \epsilon B' \Sigma^{-1}
\]
for $\epsilon>0$. Then, from \eqref{eq:lyapunov2},
\begin{eqnarray*}
(A-BK_\epsilon)\Sigma + \Sigma(A-BK_\epsilon)'&=&-\epsilon BB'-B_1B_1'\\
& \leq& -\epsilon BB'.
\end{eqnarray*}
The fact that $A-BK_\epsilon$ is Hurwitz is obvious. If now $\Sigma_\epsilon$ is the solution to
\begin{eqnarray*}
(A-BK_\epsilon)\Sigma_\epsilon + \Sigma_\epsilon (A-BK_\epsilon)'&=&-B_1B_1'
\end{eqnarray*}
the difference $\Delta=\Sigma-\Sigma_\epsilon\geq 0$ and satisfies
\begin{eqnarray*}
(A-BK_\epsilon)\Delta + \Delta(A-BK_\epsilon)'&=&-\epsilon BB',
\end{eqnarray*}
and hence is of $o(\epsilon)$.
\end{remark}

\section{Example}\label{sec:example}
Consider inertial particles that are modeled by
   \begin{eqnarray*}
       dx(t) &=& v(t)dt + dw(t)\\
       dv(t) &=& u(t)dt.
   \end{eqnarray*}
Here, $u(t)$ is the control input (force) at our disposal, $x(t)$ represents position and $v(t)$ velocity, while $w(t)$ represents random displacement due to impulsive accelerations.
The purpose of the example is to highlight a case where the control is handicapped compared to the effect of noise. Indeed, the displacement $w(t)$ is directly affecting the position while the control effort needs to be integrated before it impacts the position of the particles.

Another interesting aspect of this example is that $\cR(B)\not\subseteq \cR(B_1)$ since
$B=[0,~1]'$ while $B_1=[1,~0]'$. If we choose
\begin{equation}\label{eq:stationaryvalue}
\Sigma_1=\left[\begin{matrix}1&-1/2\\ \hspace*{2pt}-1/2& \phantom{-}1/2\end{matrix}\right]
\end{equation}
as a candidate stationary state-covariance, it can be seen that
\eqref{eq:lyapunov3} has a unique solution $X$ giving rise to  $K=\left[1,~1\right]$ and a stable feedback since $A-BK$ is Hurwitz.

We now wish to steer the spread of the particles from an initial Gaussian distribution with $\Sigma_0=2I$ at $t=0$ to the terminal marginal $\Sigma_1$
at $t=1$, and from there on, since $\Sigma_1$ is an admissible stationary state-covariance, to maintain with constant state-feedback control.

Figure~\ref{fig:Eg1Phase1} displays typical sample paths in phase space, as a function of time, that are attained using the optimal feedback strategy derived following \eqref{eq:sdp} over the time interval $[0,\,1]$.
The corresponding feedback gains $K(t)=[k_1(t),\,k_2(t)]$ are shown in Figure~\ref{fig:Eg1Controlfeedback} as functions of time.
\begin{figure}\begin{center}
\includegraphics[width=0.47\textwidth]{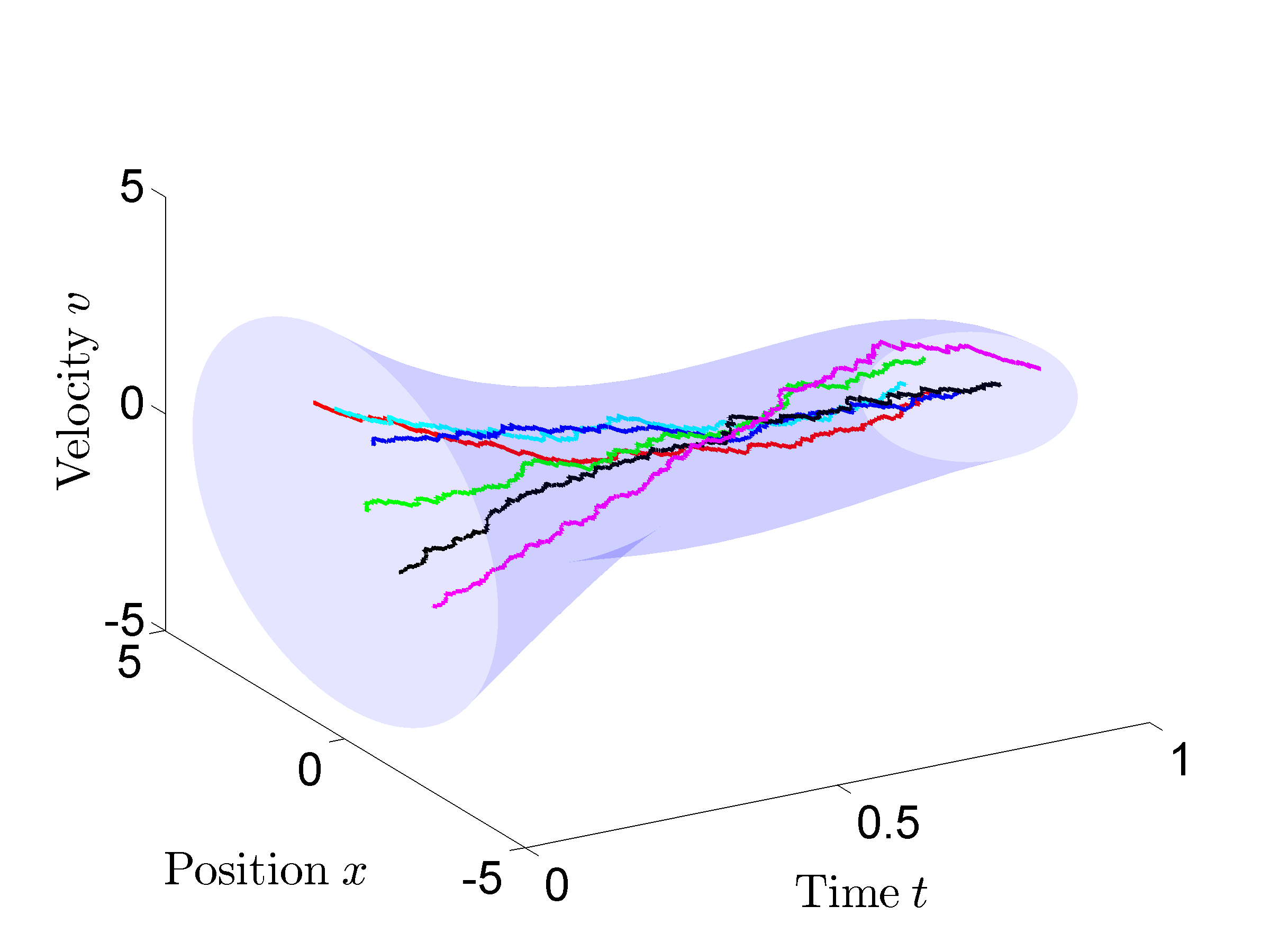}
   \caption{Inertial particles: state trajectories}
   \label{fig:Eg1Phase1}
\end{center}\end{figure}
\begin{figure}\begin{center}
\includegraphics[width=0.47\textwidth]{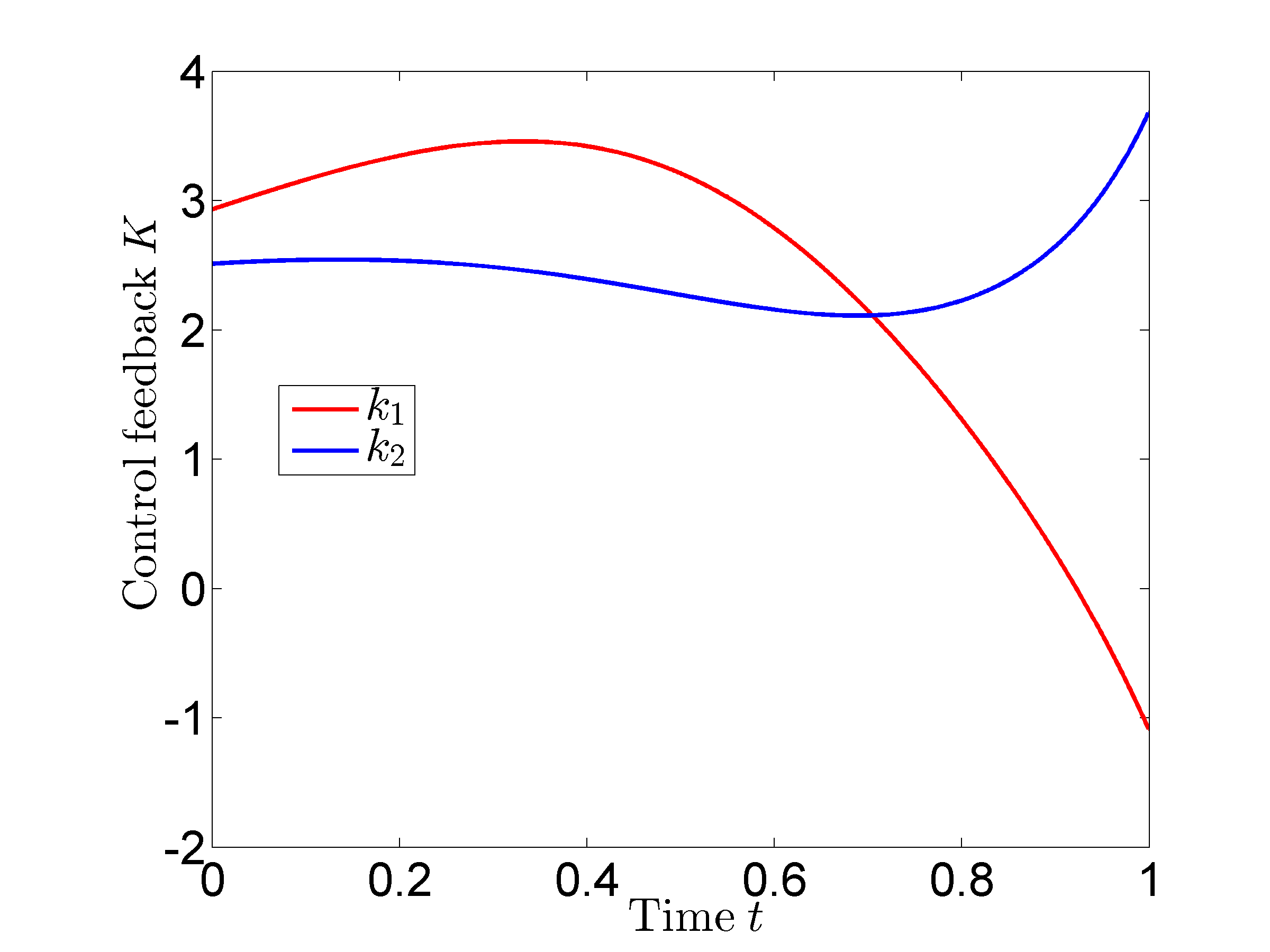}
   \caption{Inertial particles: feedback gains}
   \label{fig:Eg1Controlfeedback}
\end{center}\end{figure}
Past the point $t=1$, the state-covariance of the closed-loop system is maintained at the stationary value in \eqref{eq:stationaryvalue}.
Figure \ref{fig:Eg2Phase1} displays representative sample paths in phase space under the now constant state feedback gain $K=[1,\,1]$ over  time window $[1,\,5]$.
Finally, Figure \ref{fig:Eg2Control1} displays the corresponding control action for each trajectory over the complete time interval $[0,\,5]$, which consists of the ``transient'' interval $[0,\,1]$ to the target (stationary) distribution and the ``stationary'' interval $[1,\,5]$.
\begin{figure}\begin{center}
\includegraphics[width=0.47\textwidth]{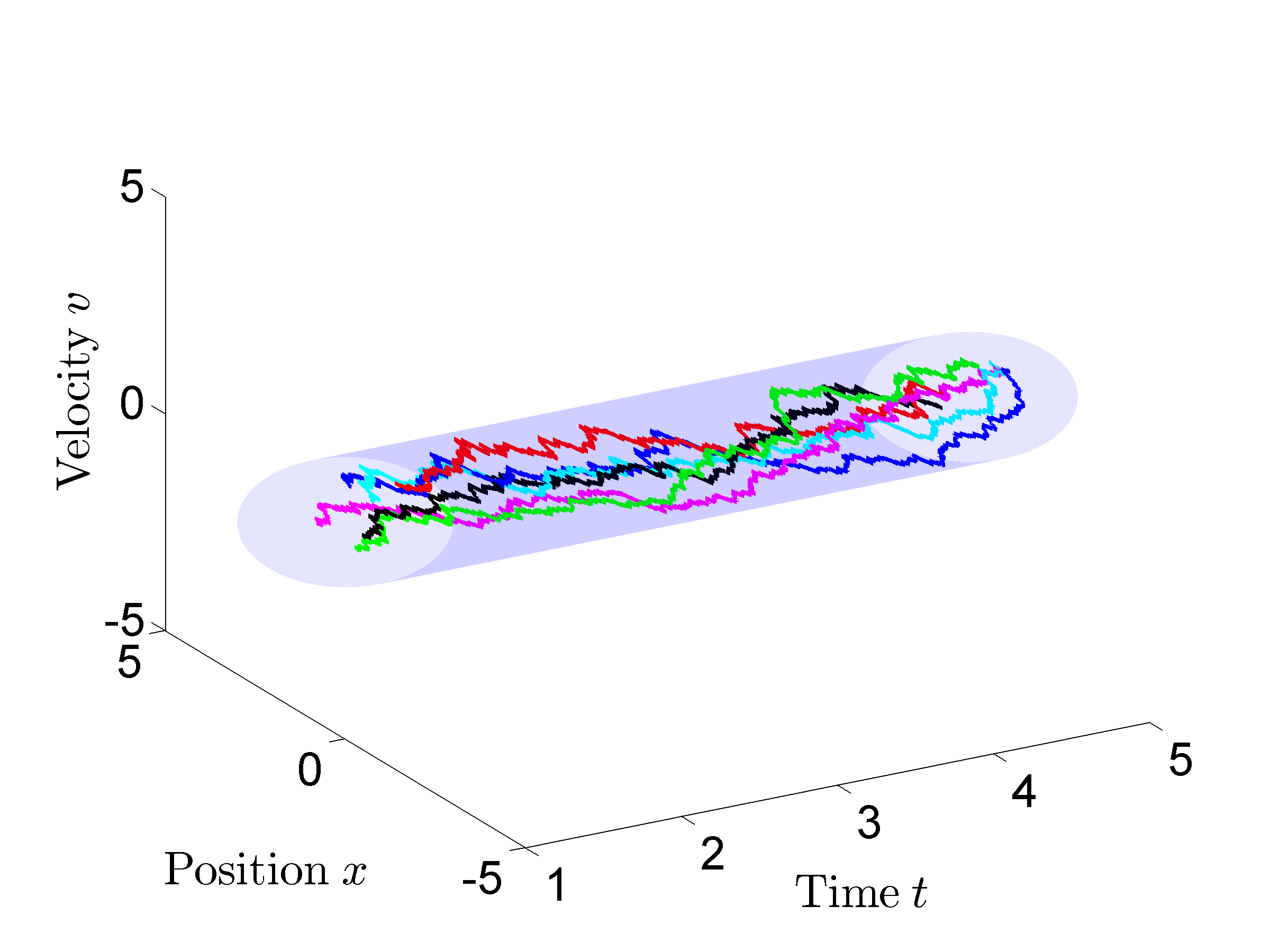}
   \caption{Inertial particles: stationary state trajectories}
   \label{fig:Eg2Phase1}
\end{center}\end{figure}
\begin{figure}\begin{center}
\includegraphics[width=0.47\textwidth]{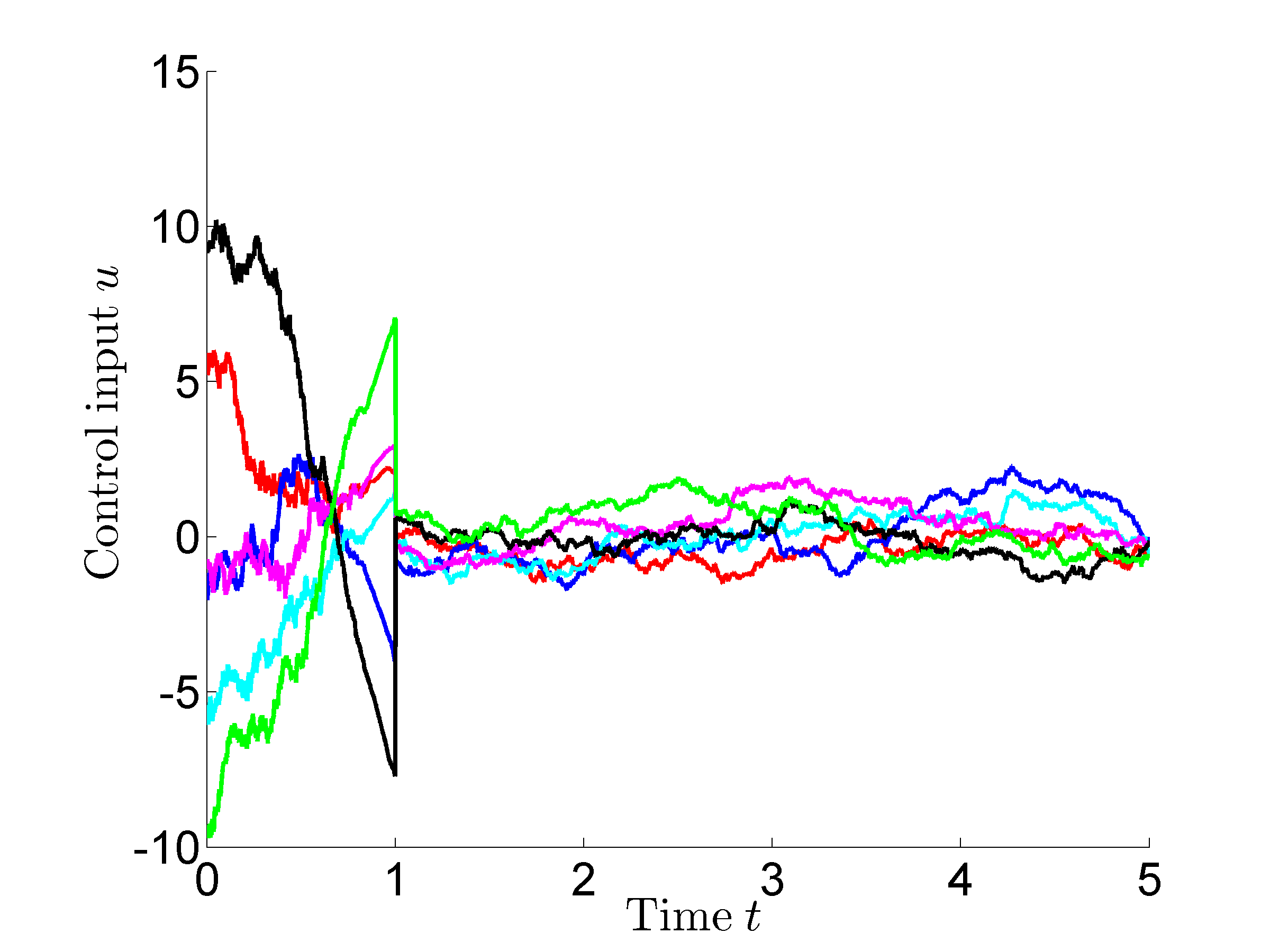}
   \caption{Inertial particles: control inputs}
   \label{fig:Eg2Control1}
\end{center}\end{figure}

\section{Appendix}\label{sec:appendix}

\begin{lemma}\label{lemma:space}Consider the maps $\f_B$ and $\g_B$ defined in (\ref{eq:g}-\ref{eq:f}). The range of $\f_B$ coincides with the null space of $\g_B$, that is,
\[\cR(\f_B)=\cN(\g_B).\]
\end{lemma}
\begin{proof} It is immediate that
\[
\cR(\f_B)\subseteq\cN(\g_B).
\]
To show equality it suffices to show that $\left(\range(\f_B)\right)^\perp\subseteq \cN(\g_B)^\perp$. To this end, consider
\[
M\in\cS_n\cap \left(\range(\f_B)\right)^\perp.
\]
Then
\[
\trace\left(M(BX+X'B')\right)=0
\]
for all $X\in\mR^{m\times n}$. Equivalently, for $Z=MB\in\mR^{n\times m}$,
$\trace(ZX)+\trace(X'Z')=0$ for all $X$. Thus,
$\trace(ZX)=0$ for all $X$ and hence $Z=0$. Since $MB=Z=0$, then $M\Pi_{\range(B)}=0$ or, equivalently, $M\Pi_{\range(B)^\perp}=M$.
Therefore $\Pi_{\range(B)^\perp}M\Pi_{\range(B)^\perp}=M$, i.e.,
 $M\in \left(\range(\g_B)\right)$. Therefore,
\[\left(\range(\f_B)\right)^\perp\subseteq \left(\range(\g_B)\right)=\cN(\g_B)^\perp
\]
since $\g_B$ is self-adjoint, which completes the proof.
\end{proof}
\vspace*{.2in}

\spacingset{.97}
\bibliographystyle{IEEEtran}
\bibliography{refs}
\end{document}